\documentclass[letterpaper, 10 pt, conference]{ieeeconf}

\IEEEoverridecommandlockouts                              % This command is only needed if 

\usepackage{tikz}
\usepackage{textcomp}

\usetikzlibrary{matrix}
\usetikzlibrary{angles,quotes}

\usepackage{graphics,graphicx} % for pdf, bitmapped graphics files
\usepackage{epsfig} % for postscript graphics files
\usepackage{mathptmx} % assumes new font selection scheme installed
\usepackage{times} % assumes new font selection scheme installed
\usepackage{amsmath} % assumes amsmath package installed
\usepackage{amssymb,nth, mathtools,dsfont}  % assumes amsmath package installed
\usepackage{amsbsy}
\usepackage{epstopdf}
\usepackage[noadjust]{cite}
\usepackage{romannum}
\usepackage{xcolor}
\usepackage{subcaption}
\usepackage{multirow}
\usepackage{comment}

\newcommand{\bmtx}{\begin{bmatrix}}
\newcommand{\emtx}{\end{bmatrix}}
\newcommand{\bsmtx}{\left[ \begin{smallmatrix}} 
\newcommand{\esmtx}{\end{smallmatrix} \right]}
\newcommand{\field}[1]{\mathbb{#1}}
\newcommand{\R}{\field{R}}
\newcommand{\RH}{\field{RH}_\infty}

\newcommand{\N}{\field{N}}

\newcommand{\Sym}{\field{S}}
\newcommand{\pjs}[1]{{\color{red}{(PJS: #1)}}}
\newcommand{\svn}[1]{{\color{blue}{(SVN: #1)}}}

\newtheorem{corollary}{Corollary}
\newtheorem{lemma}{Lemma}

\newtheorem{remark}{Remark}
\newtheorem{defin}{Definition}
\title{Integral Quadratic Constraints for Repeated ReLU}

\author{Sahel Vahedi Noori, Bin Hu, Geir Dullerud, and Peter Seiler% <-this % stops a space
	\thanks{S. Vahedi Noori and P. Seiler are with the Department of Electrical Engineering \& Computer Science, at the University of Michigan ({\tt\small sahelvn@umich.edu} and
		{\tt\small pseiler@umich.edu}). B. Hu is with the Department of Electrical and Computer Engineering at the University of Illinois at Urbana-Champaign ({\tt \small binhu7@illinois.edu}). G. Dullerud is with the Department of Electrical and Computer Engineering at the University of Minnesota Twin Cities ({\tt \small dullerud@umn.edu}).
  }
	}

\begin{document}

\maketitle

\begin{abstract}
This paper presents a new dynamic integral quadratic constraint (IQC) for the repeated Rectified Linear Unit (ReLU).  These dynamic IQCs can be used to analyze stability and induced $\ell_2$-gain performance of discrete-time, recurrent neural networks (RNNs) with ReLU activation functions. These analysis conditions can be incorporated into learning-based controller synthesis methods, which currently rely on static IQCs.  We show that our proposed dynamic IQCs for repeated ReLU form a superset of the dynamic IQCs for repeated, slope-restricted nonlinearities. We also prove that the $\ell_2$-gain bounds are nonincreasing with respect to the horizon used in the dynamic IQC filter. A numerical example using a simple (academic) RNN shows that our proposed IQCs lead to less conservative bounds than existing IQCs.

\end{abstract}

%\begin{IEEEkeywords}
% Absolute stability; Lurye system; Kalman Conjecture; Cycle behaviour. 
%\end{IEEEkeywords}

%--------------------------------------------------------------------------------
\section{Introduction}

There is increasing interest in using recurrent neural networks (RNNs) in dynamical systems and control-related applications.  Unlike feedforward architectures, RNNs contain hidden states that evolve over time, providing an internal memory that captures temporal dependencies. While these models are expressive and powerful, their nonlinear nature makes it challenging to establish rigorous stability and performance guarantees. Such guarantees are particularly important in safety-critical settings, motivating the development of stability analysis tools for systems that incorporate RNN components.

From a systems perspective, an RNN can be viewed as a Lurye-type interconnection, in which linear state dynamics are coupled with static nonlinear activation functions \cite{revay2020contracting, richardson23, hedesh2025robust}. A common approach is to use quadratic constraints (QCs) to bound the input-output behavior of the nonlinearity. Conditions for internal stability and induced gain bounds can then be derived by combining the QCs with dissipativity theory. Classical works such as \cite{Willems:71, Willems:1968} develop QC-based analysis tools for repeated slope-restricted nonlinearities.

Recent work has adapted these QC ideas for the Rectified Linear Unit (ReLU) function, which is a common choice for the activation function \cite{noori24ReLURNN, richardson23, ebihara2021stability, ebihara21}. Since these works exploit the specific structure of ReLU, they often yield tighter guarantees than those obtained from generic slope-restricted QCs. In \cite{noori25CompleteReLU}, we derived a complete set of quadratic constraints for repeated ReLU, which bounds the repeated ReLU as tight as possible up to the sign invariance inherent in quadratic forms.

To further reduce conservatism, it is beneficial to incorporate temporal structure through dynamic integral quadratic constraints (IQCs) \cite{megretski97}. This framework can capture dynamical properties through multiplier filters. A widely used class of multipliers is the Zames--Falb family, which have been studied extensively for slope-restricted nonlinearities in both scalar and repeated settings \cite{Carrasco:2016, carrasco19, safonov2000zames, kulkarni2002all, zhang2021lyapunov, fetzer17}. Table~\ref{tab:LitReview} summarizes all the tools described for QC and IQC-based stability and performance analysis.

\begin{table}[h!]
\centering
\scalebox{1.12}{
\begin{tabular}{ |c|c|c|c| } 
\hline
Method & Nonlinearity & Type & References\\ \hline
Static QCs & Slope-restricted & Repeated & \cite{Willems:71, Willems:1968} \\ \hline
Dynamic IQCs & Slope-restricted & Scalar & \cite{Carrasco:2016, carrasco19}\\ \hline
Dynamic IQCs & Slope-restricted & Repeated & \cite{safonov2000zames, kulkarni2002all, zhang2021lyapunov, fetzer17} \\ \hline
Static QCs & ReLU & Repeated & \cite{noori24ReLURNN, noori25CompleteReLU, richardson23, ebihara2021stability, ebihara21}\\ \hline
Dynamic IQCs & ReLU & Scalar & \cite{noori2025IQCs}\\ \hline
Dynamic IQCs & ReLU & Repeated & Our paper \\ \hline
\end{tabular}
}
\caption{Summary of prior methods for stability analysis of nonlinear Lurye-type interconnections.}
\label{tab:LitReview}
\end{table}

The main contribution of this paper is to develop a new class of dynamic 
IQCs for repeated ReLU.   This builds on the dynamic IQCs devloped in our prior work for scalar ReLU  \cite{noori2025IQCs}.  These IQCs exploit the structural properties of ReLU and, as shown in this paper, yield less conservative stability and performance certificates than existing slope-restricted formulations. We further show that the optimal induced-$\ell_2$ gain bounds are non-increasing with respect to the time-horizon used in the IQC filter. This implies that increasing the IQC filter horizon cannot, in theory, yield a worse (larger) bound.  A simple numerical example demonstrates the effectiveness of the proposed approach and highlights improvements over static QC and dynamic IQC methods in the literature.

%--------------------------------------------------------------------------------
\section{Notation}

This section briefly reviews basic notation regarding vectors, matrices, and signals.  Let $\R^n$ and $\R^{n\times m}$ denote the sets of real $n\times 1$
vectors and $n\times m$ matrices, respectively. Moreover, $\R_{\ge 0}^n$ and $\R_{\ge 0}^{n\times m}$ denote vectors and matrices of the given
dimensions with non-negative entries.  A matrix $M\in \R^{n\times n}$ is a \emph{Metzler matrix} if the off-diagonal entries are non-negative, i.e. $M_{ij}\ge 0$ for $i\ne j$.  A matrix $M\in \R^{n\times n}$ is \emph{doubly hyperdominant} if the off-diagonal elements are non-positive, and all row and column sums are non-negative.  $\Sym^n$ is the set of real symmetric, $n\times n$ matrices. Finally, let $\{Q_i\}_{i=-N}^N \subset \R^{m \times m}$ be given.  A \emph{block-Toeplitz} matrix $M \in \R^{m(N+1)\times m(N+1)}$ is formed from these matrices with $\{Q_0,\ldots,Q_N\}$ along the first row, $\{Q_0,\ldots,Q_{-N}\}$ along the first column, and constant blocks along each diagonal. For example, if $N=2$ then
\begin{align}
  M:= \bmtx Q_0 & Q_1 & Q_2 \\ Q_{-1} & Q_0 & Q_1 \\ Q_{-2} & Q_{-1} & Q_0 \emtx.
\end{align}

Next, let $\N$ denote the set of non-negative integers.  Let $v:\N \to \R^n$ and $w:\N \to \R^n$ be real, vector-valued sequences.  Define the inner product $\langle v,w \rangle : = \sum_{k=0}^\infty v(k)^\top w(k)$.
A sequence $v$ is said to be in $\ell_2^n$ if $\langle v,v\rangle <\infty$. In addition, the $2$-norm for $v \in \ell_2^n$ is defined as
$\|v\|_2:=\sqrt{ \langle v,v\rangle}$.   We will use $\ell_{2e}$ to denote the extended space of sequences whose $\ell_2$ norm is not necessarily finite.
$\mathbb{RL_\infty}$ denotes the set of rational functions with real coefficients that have no poles on the unit circle. $\mathbb{RH_\infty}$ is the subset of functions in $\mathbb{RL_\infty}$ with all poles inside the unit disk of the complex plane.

%A sequence $v$ belongs to the $\ell_{2}$ space if $\sum_{k=0}^\infty |v_k|^2 < \infty$. %We define $\ell^2_{\ge 0} := \left\{w\in \ell_{2e} \;\middle|\; w_n \ge 0 \text{ for all } n \in \mathbb{N} \right\}$. 

%--------------------------------------------------------------------------------
\section{Problem statement}

Consider the interconnection shown in Figure~\ref{fig:LFTdiagram}  with a static, memoryless nonlinearity $\Delta_F$ wrapped in feedback around the top channels
of a nominal system $G$.  This interconnection is denoted as $F_U(G,\Delta_F)$.  
\begin{figure}[h!t]
\centering
\begin{picture}(100,90)(40,20)
 \thicklines
 \put(75,25){\framebox(40,40){$G$}}
 \put(139,40){$d$}
 \put(150,35){\vector(-1,0){35}}  
 \put(46,40){$e$}
 \put(75,35){\vector(-1,0){35}}  
  \put(78,73){\framebox(34,34){
    $\Delta_F$}}
 \put(46,70){$v$}
 \put(55,55){\line(1,0){20}}  
 \put(55,55){\line(0,1){35}}  
 \put(55,90){\vector(1,0){23}}  
 \put(139,70){$w$}
 \put(135,90){\line(-1,0){23}}  
 \put(135,55){\line(0,1){35}}  
 \put(135,55){\vector(-1,0){20}}  
\end{picture}
\caption{Interconnection $F_U(G,\Delta_F)$ of a nominal discrete-time LTI system $G$ and a static, memoryless nonlinearity $\Delta_F$.}
\label{fig:LFTdiagram}
\end{figure}
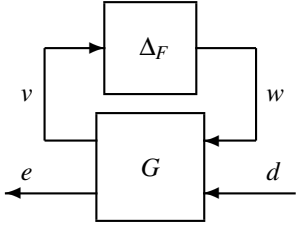

The nominal part $G$ is a discrete-time, linear time-invariant (LTI)
system described by the following state-space model:
\begin{align}
  \label{eq:LTInom}
  \begin{split}
    x(k+1) & = A\, x(k) + B_1\,w(k) +  B_2\, d(k) \\
    v(k) & =C_{1}\,x(k)+D_{11}\, w(k)+ D_{12} \,d(k)\\
    e(k) & =C_{2}\,x(k)+D_{21}\, w(k)+ D_{22}\,d(k),
  \end{split}
\end{align}
where $x(k) \in \R^{n_x}$ is the state at time $k$. Similarly, the inputs at time $k$ are $w(k) \in \R^{m}$ and $d(k)\in \R^{n_d}$ while the outputs at time $k$ are $v(k) \in \R^{m}$ and $e(k)\in \R^{n_e}$.  The  nonlinearity $\Delta_F:\ell_{2e}^m\to \ell_{2e}^m$ is defined by application of a static, memoryless function $F:\R^m \to \R^m$ at each point in time. In other words, $\Delta_F$ maps $v\in \ell_{2e}^m$ to $w=\Delta_F(v)\in\ell_{2e}^m$ by $w(k)=F( v(k) )$ at each time $k$.
The interconnection $F_U(G,\Delta_F)$ is known as a linear fractional transformation (LFT) in the robust control literature \cite{zhou96}. 

This feedback interconnection involves an implicit equation if $D_{11}\ne 0$.  Specifically, the second equation in \eqref{eq:LTInom} combined with $w(k)=F( v(k) )$ yields:
\begin{align}
   \label{eq:WellPosed}
   v(k)=C_{1}\,x(k)+D_{11}\, F( v(k) )+ D_{12} \,d(k).
\end{align}
This equation is \emph{well-posed} if there exists a unique solution $v(k)$ for all  values of $x(k)$ and $d(k)$.  Well-posedness of this equation implies that the dynamic system $F_U(G,\Delta_F)$ is well-posed in the following sense:  for all
initial conditions $x(0)\in\R^{n_x}$ and  inputs $d\in \ell_{2e}^{n_d}$ there exists unique solutions $x\in \ell_{2e}^{n_x}$, $e\in \ell_{2e}^{n_e}$ and $w, v \in \ell_{2e}^{m}$ to the system $F_U(G,\Delta_F)$. There are simple sufficient conditions for well-posedness of \eqref{eq:WellPosed}, e.g. Lemma 1 in \cite{richardson23} (which relies on results in \cite{valmorbida18,zaccarian02}). Thus, we assume well-posedness, for simplicity, in the main results of our paper.

A well-posed interconnection $F_U(G,\Delta_F)$ is \emph{internally stable} if $x(k)\to 0$ from any initial condition $x(0)$ with $d(k)=0$ for $k\in \N$. In other words, $F_U(G,\Delta_F)$ is internally stable if $x=0$ is a globally asymptotically stable equilibrium point with no external input.  A well-posed interconnection $F_U(G,\Delta_F)$ has \emph{finite induced-$\ell_2$ gain} if there exists  $\gamma<\infty$ such that the output $e$ generated by any $d\in \ell_2$ with $x(0)=0$ satisfies $\|e\|_2 \le \gamma \, \|d\|_2$.  The infimum of all such bounds on the input/output gain is denoted by $\|F_U(G,\Delta_F)\|_{2\to2}$.

The focus of this paper is on the specific case where the  nonlinearity is repeated ReLU.
Specifically, define
the scalar ReLU $\phi:\R\to \R_{\ge 0}$ as follows:
\begin{align}
\label{eq:ReLU}
\phi(v) = \left\{
  \begin{array}{ll}
    0 & \mbox{if } v < 0 \\
    v & \mbox{if } v \geq 0 
  \end{array} 
  \right. .
\end{align}
The repeated ReLU $\Phi:\R^m \to \R_{\ge 0}^m$ maps $v$ to $w$ by elementwise application of the scalar ReLU, i.e.,  $w_i=\phi(v_i)$ for $i=1,\ldots,m$.
The corresponding operator $\Delta_\Phi:\ell_{2e}^m\to \ell_{2e}^m$ is defined by applying the repeated ReLU $\Phi$ at each point in time.
The goal of this paper is to derive sufficient conditions, using the properties of the ReLU,
that verify $F_U(G,\Delta_\Phi)$ is internally stable, has finite induced-$\ell_2$ gain, and can be used to compute a bound $\gamma$.

\begin{comment}
\begin{figure}[h!t]
\centering
\begin{picture}(100,65)(170,30)
 \thicklines
  % Graph of ReLU
 \put(170,45){\vector(1,0){100}}  
 \put(260,50){$v$}
 \put(220,30){\vector(0,1){65}}  
 \put(175,85){$w=\phi(v)$}
 \put(170,46){{\color{red} \line(1,0){50}} }  
 \put(220,46){{\color{red} \line(1,1){35}} }  
\end{picture}
\caption{Graph of scalar ReLU $\phi$. \pjs{Do we need this figure now?}}
\label{fig:ScalarReLU}
\end{figure}
\end{comment}

%--------------------------------------------------------------------------------
\section{Preliminary Results}
\label{sec:Prelim}

\subsection{Integral Quadratic Constraints (IQCs)}
\label{sec:IQCs}
In this subsection, we review the definitions of static QCs and dynamic IQCs. We will define these concepts for a general static, memoryless function $F:\R^m \to \R^m$ and its corresponding operator $\Delta_F:\ell_{2e}^m \to \ell_{2e}^m$.

\vspace{0.1in}
\begin{defin}
A function $F:\R^{m} \to \R^{m}$ satisfies the \underline{Quadratic Constraint (QC)} defined by $M\in \Sym^{2m}$ if the
following inequality holds for all $v\in \R^{m}$:
\begin{align}
\label{eq:QCdef}
\bmtx v \\ F(v)\emtx^\top 
M
\bmtx v \\ F(v)\emtx \ge 0 .
\end{align}
\end{defin}
\vspace{0.1in}

QCs are, in general, conservative bounds on the graph of a function. The constraints are useful as they can be easily incorporated into stability and performance conditions for dynamical systems. These can be generalized to dynamic IQCs as formally defined next.

\vspace{0.1in}
\begin{defin}
  \label{def:tdiqc}
  Let $\Psi \in \RH^{n_r \times 2m}$ and $M\in \Sym^{n_r}$ be given. A nonlinearity $\Delta_F:\ell_{2e}^m \to \ell_{2e}^m$ defined by a  function $F: \R^m \to \R^m$ satisfies the time domain, \underline{Integral Quadratic Constraint (IQC)} defined by $(\Psi, M)$ if the following inequality holds for all $v \in \ell_{2e}$ and for all $T_0\ge 0$,
     \begin{align}
      \label{eq:tdhardiqc}
      \sum_{k=0}^{T_0} r(k)^\top M r(k)  \ge 0,
    \end{align}
  where $r$ is the output of $\Psi$ from zero initial conditions and driven by inputs $(v,w)$ where $w=\Delta_F(v)$.
\end{defin}
\vspace{0.1in}

Note that if $\Psi:= I_{2m}$ then $r=\bsmtx v \\ w \esmtx$. It follows that $\Delta_F$ satisfies the dynamic IQC with $\Psi = I_{2m}$ and $M\in\Sym^{2m}$ if and only if the function $F$ satisfies the corresponding QC with the same $M$. Definition~\ref{def:tdiqc} is called a \emph{hard} IQC because the constraint~\eqref{eq:tdhardiqc} holds for all $T_0\ge 0$.  
This is in contrast to \emph{soft} IQCs in the literature
which only hold as $T_0\to \infty$. Moreover, soft IQCs require the additional assumption that $v$, $w\in \ell_2$ to ensure the summation in  \eqref{eq:tdhardiqc} remains finite
as $T_0\to \infty$. 

\subsection{Stability and Performance Condition}
\label{sec:StabCond}

This subsection reviews a sufficient condition for internal stability and finite induced-$\ell_2$ gain for the interconnection $F_U(G,\Delta_F)$. The condition relies on dissipation inequalities and dynamic (hard) IQCs. First, partition the LTI system $G$ in \eqref{eq:LTInom} compatibly with the inputs $(w,d)$ and outputs $(v,e)$:
\begin{align*}
    G = \bmtx G_{11} & G_{12} \\ G_{21} & G_{22}\emtx .
\end{align*}
Next, define an augmented system $\hat{G}$ from the nominal system $G$ and IQC filter $\Psi$ as follows:
\begin{align}
\label{eq:Ghat}
\hat{G} & := 
\bmtx \Psi & 0 \\ 0 & I_{n_e}\emtx 
\bmtx \bmtx G_{11} & G_{12}\\ I_m & 0 \emtx \\ \\ \bmtx G_{21} & G_{22} \emtx \emtx.
\end{align}
The augmented system has $(w,d)$ as inputs and $(r,e)$ as outputs as shown in  Figure~\ref{fig:AugLFT}. The IQC $(\Psi,M)$ defines a constraint on $r$ that is satisfied by any input/output pair $(v,w)$ of $\Delta_F$. We can use the LTI dynamics of the augmented system combined with this IQC to draw conclusions about the stability and performance of the original nonlinear system $F_U(G,\Delta_F)$.

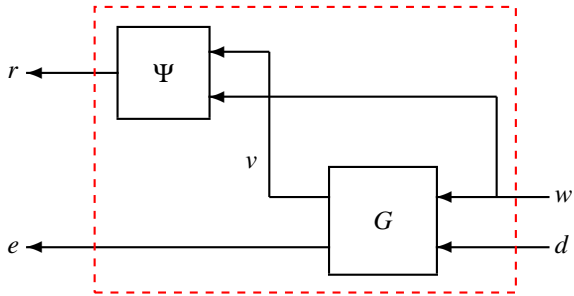
\begin{figure}[h!t]
\centering
\begin{picture}(240,112)(-62,-39)
 \thicklines
 % upper half
 \put(157,0){$w$}
 \put(135,40){\vector(-1,0){109}}  
% \put(109,40){\line(1,0){30}}  
 \put(40,13){$v$}
 \put(49,57){\vector(-1,0){23}}
% I/O for Repeated ReLU
% \put(80,75){\framebox(30,30){$\Phi$}}
% \put(75,25){\framebox(34,34){F}}
% \put(112,57){$\Phi$}
% \put(58,40){\line(0,1){45}}
% \put(58,85){\vector(1,0){90}}
% \put(125,40){\line(0,1){25}}
% \put(125,65){\vector(1,0){23}}
 \put(-8,32){\framebox(34,34){$\Psi$}}
 \put(-8,49){\vector(-1,0){35}}
 \put(-50,47){$r$}
 % lower half 
 \put(72,-27){\framebox(40,40){$G$}}
 % v
 \put(49,57){\line(0,-1){55}}
 \put(49,2){\line(1,0){23}}
 % w
 \put(135,40){\line(0,-1){38}}
 \put(155,2){\vector(-1,0){43}}
 % d
 \put(157,-19){$d$}
 \put(155,-17){\vector(-1,0){43}}
 %e
 \put(-50,-19){$e$}
 \put(72,-17){\vector(-1,0){115}}
 % dotted box
 \begin{tikzpicture}[overlay]
    \draw[red,thick,dashed](-0.6,-1.2) -- (5,-1.2) -- (5,2.6) -- (-0.6,2.6) -- cycle;;
 \end{tikzpicture}
\end{picture}
\caption{Augmented LTI system $\hat{G}$ mapping $(w,d)$ to $(r,e)$ based on the interconnection of the nominal system $G$ and IQC filter $\Psi$.}
\label{fig:AugLFT}
\end{figure}

The stability and performance condition is given in terms of a state-space model for $\hat{G}$:
\begin{align}
  \label{eq:LTIAugmented}
  \begin{split}
    \hat{x}(k+1) & = \hat{A} \, \hat{x}(k) + \hat{B}_1\,w(k) 
       +  \hat{B}_2 \, d(k) \\
    r(k) & = \hat{C}_1 \, \hat{x}(k) + \hat{D}_{11}\, w(k) + \hat{D}_{12} \,d(k) \\
    e(k) & = \hat{C}_2\, \hat{x}(k) + \hat{D}_{21}\, w(k) + \hat{D}_{22}\, d(k).
  \end{split}
\end{align}
We can express the state matrices of $\hat{G}$ in terms of the state matrices of the nominal model $G$ and IQC filter $\Psi$.
Explicit expressions are given in the appendix.  The state-dimension of the augmented system, denoted $n_{\hat{x}}$, is the sum of the state dimensions of $G$ and $\Psi$.

Next, define the following matrix function 
using the state matrices of $\hat{G}$:
\begin{align}
\label{eq:LiftedLMI}
\begin{split}
 L(P,M,\gamma^2) &:= \bmtx \hat{A}^\top P \hat{A}-P  & \hat{A}^\top P \hat{B}_1 
  &  \hat{A}^\top P \hat{B}_2 \\ 
  \hat{B}_1^\top P \hat{A} & \hat{B}_1^\top P \hat{B}_1 
  & \hat{B}_1^\top P \hat{B}_2 \\
  \hat{B}_2^\top P \hat{A} & \hat{B}_2^\top P \hat{B}_1 
& \hat{B}_2^\top P \hat{B}_2-\gamma^2 I\emtx \\
& + \bmtx \hat{C}_2^\top \\ \hat{D}_{21}^\top \\ \hat{D}_{22}^\top\emtx
  \bmtx \hat{C}_2^\top \\ \hat{D}_{21}^\top \\ \hat{D}_{22}^\top \emtx^\top
+  \bmtx \hat{C}_1^\top \\ \hat{D}_{11}^\top  \\ \hat{D}_{12}^\top \emtx
M
  \bmtx \hat{C}_1^\top \\ \hat{D}_{11}^\top  \\ \hat{D}_{12}^\top \emtx^\top
\end{split}
\end{align}
The lemma below provides a sufficient condition for internal stability and finite induced-$\ell_2$ gain of the interconnection $F_U(G,\Delta_F)$.

\vspace{0.1in}
\begin{lemma}
    \label{lem:stabGeneral}
    Consider the interconnection $F_U(G,\Delta_F)$ where
    $G$ is the LTI system \eqref{eq:LTInom} and $\Delta_F:\ell_2^m \to \ell_2^m$ is a static, memoryless nonlinearity that satisfies the time-domain IQC defined by $(\Psi,M)$.  Assume the interconnection is well-posed.
    Then the interconnection $F_U(G,\Delta_F)$ is internally stable and has $\|F_U(G,\Delta_F)\|_{2\to 2} < \gamma$ if there exists a 
    $P\succeq 0$ and  $\gamma \ge 0$ such that $L(P,M,\gamma^2) \prec 0$.
\end{lemma}
\vspace{0.1in}
\begin{proof} 
By well-posedness, for any $x(0) \in \R^{n_x}$ there is a unique solution $(x,v,e,w)$ to the interconnection $F_U(G,\Delta_F)$.  Moreover, the filter $\Psi$ driven by $(v,w)$ from zero initial conditions has a unique output $r$. Let $\hat{x}$ denote the corresponding state trajectory for the augmented system $\hat{G}$.

The matrix inequality $L(P,M,\gamma^2) \prec 0$ is strictly feasible. Hence it remains feasible under small perturbations: for a sufficiently small $\epsilon>0$ the perturbed matrix $\hat{P}:= P+\epsilon I \succ 0$ satisfies $L(\hat{P},M,\gamma^2) + \epsilon I \preceq 0$.

Define a storage function by $V\left(\hat{x}\right) := \hat{x}^\top \hat{P} \hat{x}$.  Left and right multiply the perturbed LMI by $[\hat{x}(k)^\top \; w(k)^\top \; d(k)^\top]$ and its transpose.  The result, applying the augmented dynamics~\eqref{eq:Ghat}, gives the following:
\begin{align*}
    \begin{split}  
    & V(\hat{x}(k+1)) - V(\hat{x}(k)) + \epsilon \hat{x}(k)^\top \hat{x}(k) + e(k)^\top e(k) \\
    & + r(k)^\top M r(k) \leq (\gamma^2-\epsilon)d(k)^\top d(k) - \epsilon w(k)^\top w(k)
    \end{split}
\end{align*}
Summing the inequality from $k = 0$ to an arbitrary time $k = T_0 \ge 0$ yields:
\begin{align*}
    & V(\hat{x}(T_0+1)) - V(\hat{x}(0)) + \epsilon \sum_{k=0}^{T_0} \hat{x}(k)^\top \hat{x}(k) + \sum_{k=0}^{T_0} e(k)^\top e(k) \\
    & + \sum_{k=0}^{T_0} r(k)^\top M r(k) \leq
    \sum_{k=0}^{T_0}
    (\gamma^2-\epsilon)  d(k)^\top d(k) -\epsilon \sum_{k=0}^{T_0} w(k)^\top w(k)
\end{align*}
The first term is nonnegative since $\hat{P}\succ 0$. The 
sum involving $r(k)^\top M r(k)$
is also nonnegative because $\Delta_F$ satisfies the
IQC  defined by $(\Psi,M)$.  Hence this implies
\begin{align}
\label{eq:sumDissipation}
    \begin{split}
        & \epsilon \sum_{k=0}^{T_0} \hat{x}(k)^\top \hat{x}(k) + \sum_{k=0}^{T_0} e(k)^\top e(k) \\
        & \leq V(\hat{x}(0)) + (\gamma^2-\epsilon) \sum_{k=0}^{T_0} d(k)^\top d(k)
    \end{split}
\end{align}
Note that the inequality still holds even though the term $-\epsilon \sum_{k=0}^{T_0} w(k)^\top w(k) \le 0$ has been dropped.

First consider $d(k)=0$ for all $k$. Then \eqref{eq:sumDissipation} simplifies to the following inequality:
\begin{align*}
    \epsilon \sum_{k=0}^{T_0} \hat{x}(k)^\top \hat{x}(k) \le V(\hat{x}(0))
\end{align*}
This bound holds for any $T_0 \ge 0$. Take the limit as $T_0\to \infty$ to conclude that $\| \hat{x}\|_2^2 \le \frac{1}{\epsilon} V(\hat{x}(0))$.  It follows that the $\ell_2$-norm of $\hat{x}$ is bounded since $V(\hat{x}(0))$ is finite. This implies $\hat{x}(k) \to 0$ as $k\to \infty$ and hence $x(k) \to 0$ as well. We conclude that the interconnection $F_U(G,\Delta_F)$ is internally stable.

Now assume $x(0)=0$ and $d \in \ell_{2}$. Inequality \eqref{eq:sumDissipation}
combined with $V(\hat{x}(0))=0$
yields:
\begin{align*}
    \sum_{k=0}^{T_0} e(k)^\top e(k) \leq (\gamma^2-\epsilon) \sum_{k=0}^{T_0} d(k)^\top d(k)
\end{align*}
Again, take the limit as $T_0\to \infty$ to conclude that $\| e\|_2^2 \le (\gamma^2-\epsilon) \| d\|_2^2$. 
We conclude that  the interconnection has finite induced-$\ell_2$ gain with $\|F_U(G,\Delta_F)\|_{2\to 2} < \gamma$.
\end{proof}

%--------------------------------------------------------------------------------
\section{Main Results}
In this subsection, we first review dynamic IQCs for repeated slope-restricted nonlinearities, which serve as the starting point for the results developed in the following subsections. We then consider the repeated ReLU nonlinearity $\Phi:\R^{m} \to \R_{\ge 0}^{m}$ where $\phi:\R\to\R_{\ge 0}$ denotes the scalar ReLU, and derive a new dynamic IQC for this nonlinearity.

\subsection{IQCs for Repeated Slope Restricted Nonlinearities}
\label{sec:SlopeResIQCs}

In this subsection, we review a known dynamic IQC for repeated, slope restricted functions.  Let $f:\R \to \R$ be a scalar function with $f(0)=0$. Moreover, assume $f$ has slope restricted to $[0,1]$. This corresponds to the following constraint:
\begin{align}
  0 \le  \frac{ f(v_2)-f(v_1)}{v_2-v_1} \le 1, 
  \,\, \forall v_1, \, v_2 \in \R, \, v_1\ne v_2.
\end{align}
The repeated nonlinearity $F:\R^m \to \R^m$ is defined by $w=F(v)$
with elementwise application of the scalar function $f$.
The next lemma defines a standard static QC for this class of  repeated slope restricted functions. 

\vspace{0.1in}
\begin{lemma}
\label{lem:doublyhypQC}
Let $F:\R^{m} \to \R^{m}$ be a repeated nonlinearity defined element-wise by $f: \R \to \R$ where $f(0)=0$ and $f$ is slope restricted to $[0,1]$. Moreover, let $Q_0 \in \R^{m \times m}$ be any doubly hyperdominant matrix and define
\begin{align}
    M:= \bmtx 0 & Q_0^\top \\ Q_0 & -(Q_0+Q_0^\top) \emtx.
\end{align}
Then the following QC holds  $\forall v\in\R^{m}$ and $w=F(v)$:
\begin{align}
\label{eq:doublyhypQC}
    \bmtx v \\ w \emtx^\top
     M   
    \bmtx v \\ w \emtx \ge 0
\end{align}
\end{lemma}
\vspace{0.1in}
\begin{proof}
This lemma follows from the results in Section 3.5 of \cite{Willems:71}.
\end{proof}
\vspace{0.1in}

The corresponding operator $\Delta_F:\ell_{2e}^m\to \ell_{2e}^m$ is again defined by application of $F$ pointwise in time. Next, we define a dynamic IQC for $\Delta_F$ using the following dynamic filter.
\begin{align}
    \label{eq:FiniteFilter}
    &\Psi_N(z) := 
\\
\nonumber
    & \left[\begin{array}{cccc|cccc}
        I_{m} & z^{-1}I_{m}  & \cdots & z^{-N}I_{m} & 0 & 0 & \cdots & 0\\
        \hline
        0 & 0 & \cdots & 0 & I_{m} & z^{-1}I_{m} & \cdots & z^{-N}I_{m}
        \end{array}\right]^\top.
\end{align}
This finite impulse response filter has dimension
$2m(N+1)\times 2m$. It
has been used previously for numerical implementations of discrete-time Zames--Falb IQCs \cite{carrasco19}. Exciting this filter with an input sequence $\bsmtx v \\ w \esmtx$ from zero initial conditions gives the following output at time $k$:
\begin{align}
\label{eq:rN}
    r(k) := 
    &\bmtx v(k)^\top, \ldots, v(k-N)^\top, \, w(k)^\top, \ldots, w(k-N)^\top \emtx^\top %\in \R^{(2N+2)m}.
\end{align}
The assumption of zero initial conditions implies that $v(j)=0$ and $w(j)=0$ for all $j<0$.

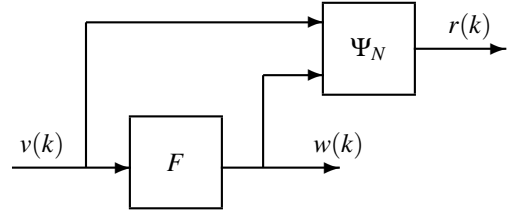
\begin{figure}[h!t]
\centering
\begin{picture}(240,75)(0,20)
 \thicklines
 \put(144,46){$w(k)$}
 \put(109,40){\vector(1,0){45}}  
 \put(33,46){$v(k)$}
 \put(30,40){\vector(1,0){45}}  
% I/O for Repeated ReLU
% \put(80,75){\framebox(30,30){$\Phi$}}
 \put(75,25){\framebox(34,34){$F$}}
%\put(112,57){$\Phi$}
 \put(58,40){\line(0,1){55}}
 \put(58,95){\vector(1,0){90}}
 \put(125,40){\line(0,1){35}}
 \put(125,75){\vector(1,0){23}}
 \put(148,68){\framebox(34,34){$\Psi_N$}}
 \put(182,85){\vector(1,0){35}}
 \put(195,91){$r(k)$}
\end{picture}
\caption{Time domain graphical interpretation for filter $\Psi_N$.}
\label{fig:filter}
\end{figure}

The next lemma is a known result with similar statements appearing in \cite{fetzer17, zhang2021lyapunov}.  The proof is given as it provides the starting point for our main result in the next section.

\vspace{0.1in}
\begin{lemma}
\label{lem:dynIQCSlopeMIMO}
Let $\Delta_F:\ell_2^m \to \ell_2^m$ be a repeated nonlinearity defined by the scalar function $f: \R \to \R$, where $f(0)=0$ and $f$ has slope restricted to $[0,1]$.   Moreover, let $\{Q_i\}_{i=-N}^N \subset \R^{m \times m}$ be given and define
\begin{align*}
    M_{row}& :=\bmtx Q_{-N} & \cdots & Q_{-1} & Q_0 & Q_1 & \cdots & Q_N \emtx \\
    M_{col}& :=\bmtx Q_{-N}^\top & \cdots & Q_{-1}^\top & Q_0^\top & Q_1^\top & \cdots & Q_N^\top \emtx ^\top.
\end{align*}
Assume: (i) all entries of $Q_i$ ($i\ne 0$) are non-positive, (ii) all off-diagonal entries of $Q_0$ are non-positive, (iii) all rows sums of $M_{row}$ are non-negative, and (iv) all column sums of $M_{col}$ are non-negative.  Then  $\Delta_F$ satisfies the time-domain IQC defined by $(\Psi_N,M)$ where
\begin{align}
    \label{eq:M_SlopeResMIMO}
    M := \bmtx 0 & M_0^\top \\ M_0 & -(M_0+M_0^\top) \emtx,
\end{align}
with
\begin{align}
    \label{eq:M0MIMO}
    M_0 := 
    \bmtx Q_0 & Q_1 & \cdots & Q_N \\
    Q_{-1} & 0 & \cdots & 0\\
    \vdots \\
    Q_{-N} & 0 & \cdots & 0 \emtx.
\end{align}
\end{lemma}
\vspace{0.1in}
\begin{proof}
For any $T_0\ge 0$, define the following stacked vectors of the signals $v$ and $w=\Delta_F(v)$ from $t=0$ through $t=T_0$,
\begin{align}
    \label{eq:stackvecs1}
    \begin{split}
        \bar{v}_{T_0} & := \bmtx v(T_0)^\top & v(T_0-1)^\top & \cdots & v(0)^\top \emtx^\top \in \R^{(T_0+1)m}, \\
        \bar{w}_{T_0} & := \bmtx w(T_0)^\top & w(T_0-1)^\top & \cdots & w(0)^\top \emtx^\top \in \R^{(T_0+1)m}.
    \end{split}
\end{align}

Next, define a block-Toeplitz matrix $\bar{Q}_{T_0} \in \R^{((T_0+1)m)\times ((T_0+1)m)}$ formed from $\{Q_i\}_{i=-N}^N$.
If $T_0\le N$ then this matrix is
\begin{align*}
    \bar{Q}_{T_0} := 
    \bmtx Q_0 & Q_1 & \cdots & Q_{T_0} \\
    Q_{-1} & Q_0 & \cdots & Q_{T_0-1} \\
    \vdots & &  \ddots \\
    Q_{-T_0} & Q_{-T_0+1} & \cdots & Q_0 \emtx, 
\end{align*}
If $T_0>N$ then the first block row and column of $\bar{Q}_{T_0}$ are padded with $T_0-N$ blocks of $m\times m$ zero matrices, e.g., the first block row is: 
\begin{align*}
    \bmtx 
    Q_0 & Q_1 & \cdots & Q_N & 0 &\cdots & 0 \emtx \in \R^{m\times ((T_0+1)m)}, 
\end{align*}
In any case, the assumptions on $\{Q_i\}_{i=-N}^N$ imply that $\bar{Q}_{T_0}$ is a doubly hyperdominant matrix for any $T_0\ge 0$.

\begin{comment}
Next, define a block-Toeplitz matrix $\bar{Q}_{T_0} \in \R^{((T_0+1)m)\times ((T_0+1)m)}$ with the first $m$ rows and $m$ columns defined by:
\begin{align*}
    &\bar{Q}_{T_0}(1:m,:) \\
    &:= \bmtx 
    Q_0 & Q_1 & \cdots & Q_N & 0 &\cdots & 0 \emtx \in \R^{m\times ((T_0+1)m)} \\
    &\bar{Q}_{T_0}(:,1:m) \\
    &:=\bmtx Q_0^\top & Q_{-1}^\top & \cdots & Q_{-N}^\top & 0 & \cdots & 0\emtx^\top \in \R^{((T_0+1)m)\times m}
\end{align*}
\svn{Is the definition above clear enough or should I change it? As we discussed in the meeting I can define something like BlockToeplitz(row,column)=the big matrix I commented out}
\begin{align*}
\bar{Q}_{T_0} := \bmtx 
    Q_0 & Q_1 & \cdots & Q_N & 0 &\cdots & 0 \\
    Q_{-1} & Q_0 & \cdots & Q_{N-1} & Q_N &\cdots & 0 \\
    \vdots & \ddots & \ddots & & & \ddots \\
    Q_{-N} & \cdots & \cdots & Q_0 & Q_1 &\cdots & 0 \\
    0 & Q_{-N} & \cdots & Q_{-1} & Q_0 &\cdots & 0 \\
    \vdots & & \ddots & & \ddots & \ddots &  \\
    0 & 0 & \cdots & 0 & 0 & \cdots & Q_0
    \emtx.
\end{align*}
If $T_0\le N$ then the first block $mT_0$ rows and columns of $\bar{Q}_{T_0}$ are defined using only the sub-matrices $\{Q_i\}_{i=-T_0}^{T_0}$ (and the additional zeros will not appear).  In any case, the assumptions on $\{Q_i\}_{i=-N}^N$ imply that $\bar{Q}_{T_0}$ is a doubly hyperdominant matrix for any $T_0\ge 0$.
\end{comment}

Finally, it can be shown, by direct substitution, that the
output $r$ of $\Psi_N$ driven by $(v,w)$ from zero
initial conditions satisfies the following for any
$T_0\ge 0$: 
\begin{align*}
  \sum_{k=0}^{T_0} r(k)^\top M r(k)  
  = \bmtx \bar{v}_{T_0}  \\ \bar{w}_{T_0} \emtx^\top
    \bmtx 0 & \bar{Q}_{T_0}^\top \\ \bar{Q}_{T_0} & -(\bar{Q}_{T_0}+\bar{Q}_{T_0}^\top) \emtx
    \bmtx \bar{v}_{T_0}  \\ \bar{w}_{T_0} \emtx
\end{align*}
This is nonnegative for any $T_0\ge 0$ by Lemma~\ref{lem:doublyhypQC}.
%from the results in Section 3.5 of \cite{Willems:71}.
\end{proof}
\vspace{0.1in}

We now establish a connection between Lemma~\ref{lem:dynIQCSlopeMIMO} and prior results in the literature, including static QCs \cite{Willems:71, Willems:1968} and dynamic IQCs for scalar slope-restricted nonlinearities \cite{megretski97, Carrasco:2016, fetzer17, carrasco19}. If $N = 0$, then Lemma~\ref{lem:dynIQCSlopeMIMO} reduces to the static QC formulation in Lemma~\ref{lem:doublyhypQC}. 
If $m = 1$, then all matrices $\{Q_i\}_{i=-N}^N  \subset \R^{m \times m}$ in Lemma~\ref{lem:dynIQCSlopeMIMO} are simply scalars $\{q_i\}_{i=-N}^N$. This case corresponds to the results in \cite{carrasco19}.

\subsection{Integral Quadratic Constraints for ReLU}
\label{sec:ReLUIQCs}
The next two lemmas describe the static QC and dynamic IQC for repeated ReLU, respectively.
\vspace{0.1in}
\begin{lemma}
    \label{lem:staticQCsReLU}
    Let $\Phi:\R^{m} \to \R^{m}$ be the repeated ReLU where $\phi: \R \to \R$ is the scalar ReLU.
    Let $Q_1=Q_1^\top$, $Q_2=Q_2^\top\in \R_{\ge 0}^{m \times m}$, and $Q_3\in \R^{m \times m}$ be given with $Q_3$ a Metzler matrix. Define
    \begin{align}
        M:=  \bmtx Q_1 & -Q^\top_3-Q_1 \\ -Q_3-Q_1 & Q_1+Q_2+Q_3+Q_3^\top\emtx.
    \end{align}
    Then the following QC holds  $\forall v\in\R^{m}$ and $w=\Phi(v)$:
    \begin{align}
    \label{eq:repReLUFin}
    \bmtx v \\ w\emtx^\top M
    \bmtx v \\ w\emtx \ge 0
    \end{align}
\end{lemma}
\vspace{0.1in}
\begin{proof}
This follows from the QC results for ReLU in  \cite{richardson23, noori24ReLURNN, ebihara21}.
\end{proof}
\vspace{0.1in}

The next lemma is the primary contribution of this paper. It generalizes the dynamic IQCs previously developed for scalar ReLU \cite{noori2025IQCs} to case of repeated ReLU.

\vspace{0.1in}
\begin{lemma}
\label{lem:dynIQCReLUMIMO}
Let $\Delta_\Phi:\ell_2^m \to \ell_2^m$ be a repeated nonlinearity defined by the scalar ReLU $\phi: \R \to \R$.   Moreover, let $\{Q^1_i\}_{i=0}^N, \{Q^2_i\}_{i=0}^N, \{Q^3_i\}_{i=-N}^{-1}, \{Q^3_i\}_{i=1}^N \subset \R_{\ge 0}^{m \times m}$ be given with $Q^3_0 \in \R^{m \times m}$ a Metzler matrix. Then $\Delta_\Phi$ satisfies the time-domain hard IQC defined by $(\Psi_N,M)$ where
\begin{align}
\label{eq:M_ReLUMIMO}
   M:= \bmtx M_1 & -M^\top_3-M_1 \\ -M_3-M_1 & M_1+M_2+M_3+M_3^\top\emtx,
\end{align}
with
\begin{align}
    \label{eq:M12MIMO}
    M_j := 
    \bmtx Q^j_0 & Q^j_1 & \cdots & Q^j_N \\
    Q^j_1 & 0 & \cdots & 0\\
    \vdots \\
    Q^j_N & 0 & \cdots & 0 \emtx, \;\; j=1,2,
\end{align}
\begin{align}
    \label{eq:M3MIMO}
    M_3 := 
    \bmtx Q^3_0 & Q^3_1 & \cdots & Q^3_N \\
    Q^3_{-1} & 0 & \cdots & 0\\
    \vdots \\
    Q^3_{-N} & 0 & \cdots & 0 \emtx.
\end{align}

\end{lemma}
\vspace{0.1in}
\begin{proof}
The proof follows a similar approach to that of Lemma~\ref{lem:dynIQCSlopeMIMO}. Define block-Toeplitz matrices $\bar{Q}_{T_0}^1, \bar{Q}_{T_0}^2, \bar{Q}_{T_0}^3 \in \R^{((T_0+1)m)\times ((T_0+1)m)}$ formed from $\{Q_i^j\}_{i=-N}^N$. If $T_0\le N$ then these matrices are
\begin{comment}
\begin{align*}
    &\bar{Q}_{T_0}^j(1:m,:) = \bar{Q}_{T_0}^j(:,1:m)^\top \\
    &:= \bmtx 
    Q_0^j & Q_1^j & \cdots & Q_N^j & 0 &\cdots & 0 \emtx \in \R^{m\times ((T_0+1)m)} \\
    &\text{for $j=1,2$, and} \\
    &\bar{Q}_{T_0}^3(1:m,:) \\
    &:= \bmtx 
    Q_0^3 & Q_1^3 & \cdots & Q_N^3 & 0 &\cdots & 0 \emtx \in \R^{m\times ((T_0+1)m)} \\
    &\bar{Q}_{T_0}^3(:,1:m) \\
    &:=\bmtx Q_0^{3\top} & Q_{-1}^{3\top} & \cdots & Q_{-N}^{3\top} & 0 & \cdots & 0\emtx^\top \in \R^{((T_0+1)m)\times m}    
\end{align*}
\end{comment}
\begin{align*}
\begin{split}
    \bar{Q}_{T_0}^j &:= \bmtx Q_0^j & Q_1^j & \cdots & Q_{T_0}^j \\
    Q_1^j & Q_0^j & \cdots & Q_{T_0-1}^j \\
    \vdots & &  \ddots \\
    Q_{T_0}^j & Q_{T_0-1}^j & \cdots & Q_0^j \emtx, \;\; j=1,2, \\
    \bar{Q}_{T_0}^3 &:= \bmtx Q_0^3 & Q_1^3 & \cdots & Q_{T_0}^3 \\
    Q_{-1}^3 & Q_0^3 & \cdots & Q_{T_0-1}^3 \\
    \vdots & &  \ddots \\
    Q_{-T_0}^3 & Q_{-T_0+1}^3 & \cdots & Q_0^3 \emtx.
\end{split}
\end{align*}
If $T_0>N$ then the first block row and column of $\bar{Q}_{T_0}^j$ are padded with $T_0-N$ blocks of $m\times m$ zero matrices, e.g., the first block row is: 
\begin{align*}
    \bmtx 
    Q_0^j & Q_1^j & \cdots & Q_N^j & 0 &\cdots & 0 \emtx \in \R^{m\times ((T_0+1)m)}, 
\end{align*}
In any case, the assumptions on $\{Q_i^1\}_{i=0}^N$, $\{Q_i^2\}_{i=0}^N$, $\{Q_i^3\}_{i=-N}^N$ imply that $\bar{Q}_{T_0}^1=\bar{Q}_{T_0}^{1\top}$, $\bar{Q}_{T_0}^2=\bar{Q}_{T_0}^{2\top}\in \R_{\ge 0}^{m \times m}$, and $\bar{Q}_{T_0}^3\in \R^{m \times m}$ is a Metzler matrix.

Finally, it can be shown, by direct substitution, that the
output $r$ of $\Psi_N$ driven by $(v,w)$ from zero
initial conditions satisfies the following for any
$T_0\ge 0$: 
\begin{align*}
\begin{split}
  &\sum_{k=0}^{T_0} r(k)^\top M r(k) = \\
  &\bmtx \bar{v}_{T_0}  \\ \bar{w}_{T_0} \emtx^\top
    \bmtx \bar{Q}_{T_0}^1 & -\bar{Q}^{3\top}_{T_0}-\bar{Q}_{T_0}^1 \\ -\bar{Q}_{T_0}^3-\bar{Q}_{T_0}^1 & \bar{Q}_{T_0}^1+\bar{Q}_{T_0}^2+\bar{Q}_{T_0}^3+\bar{Q}_{T_0}^{3\top}\emtx
    \bmtx \bar{v}_{T_0}  \\ \bar{w}_{T_0} \emtx
\end{split}
\end{align*}
This is nonnegative for any $T_0\ge 0$ by Lemma~\ref{lem:staticQCsReLU}.
\end{proof}
\vspace{0.1in}

If $N = 0$, then Lemma~\ref{lem:dynIQCReLUMIMO} reduces to the static QC formulation in Lemma~\ref{lem:staticQCsReLU}. 
If $m = 1$, then all matrices $\{Q^1_i\}_{i=0}^N, \{Q^2_i\}_{i=0}^N  \subset \R_{\ge 0}^{m \times m}$ and $\{Q^3_i\}_{i=-N}^N  \subset \R^{m \times m}$ in Lemma~\ref{lem:dynIQCReLUMIMO} reduce to scalars $\{q^1_i\}_{i=0}^N, \{q^2_i\}_{i=0}^N \subset \R_{\ge 0}$ and $\{q^3_i\}_{i=-N}^N \subset \R$ respectively, leading to the result in \cite{noori2025IQCs}. 

\begin{comment}
\vspace{0.1in}
\begin{lemma}
\label{lem:dynIQCReLUSISO}
Let $\phi:\ell_{2e} \to \ell_{2e}$ be an element-wise ReLU.   Moreover, let $\{q^1_i\}_{i=0}^N, \{q^2_i\}_{i=0}^N \subset \R_{\ge 0}$, and $\{q^3_i\}_{i=-N}^N \subset \R$ be given with  $q^3_i \ge 0 $ for $i\ne 0$. Then $\phi$ satisfies the time-domain hard IQC defined by $(\Psi_N,M)$ where
\begin{align}
\label{eq:M_ReLUSISO}
   M:= \bmtx M_1 & -M^\top_3-M_1 \\ -M_3-M_1 & M_1+M_2+M_3+M_3^\top\emtx,
\end{align}
with
\begin{align}
    \label{eq:M12SISO}
    M_j := 
    \bmtx q^j_0 & q^j_1 & \cdots & q^j_N \\
    q^j_1 & 0 & \cdots & 0\\
    \vdots \\
    q^j_N & 0 & \cdots & 0 \emtx, \;\; j=1,2,
\end{align}
\begin{align}
    \label{eq:M3SISO}
    M_3 := 
    \bmtx q^3_0 & q^3_1 & \cdots & q^3_N \\
    q^3_{-1} & 0 & \cdots & 0\\
    \vdots \\
    q^3_{-N} & 0 & \cdots & 0 \emtx.
\end{align}
\end{lemma}
\vspace{0.1in}
\begin{proof}
This follows from the results in Subsection V-B of \svn{Our ACC paper}.
\end{proof}
\vspace{0.1in}
\end{comment}

\subsection{Stability Condition}
\label{sec:ReLUStab}

This subsection derives a sufficient condition for internal stability of interconnection $F_U(G,\Phi)$ with $\Phi:\ell_2^m \to \ell_2^m$ being a repeated ReLU, based on the hard IQC derived in the previous subsection. Define the transfer function $G$, filter $\Psi_N(z)$, the augmented system $\hat{G}$, and the linear matrix inequality $L(P,M,\gamma^2)$ as in Section~\ref{sec:Prelim}.

\vspace{0.1in}
\begin{corollary}
    \label{cor:StabReLU}
    Consider the interconnection $F_U(G,\Phi)$ where
    $G$ is the LTI system \eqref{eq:LTInom}.  $\Phi:\ell_2^m \to \ell_2^m$ is a repeated ReLU.  Assume the interconnection is well-posed.
    Let $\{Q^1_i\}_{i=0}^N, \{Q^2_i\}_{i=0}^N, \{Q^3_i\}_{i=-N}^{-1}, \{Q^3_i\}_{i=1}^N \subset \R_{\ge 0}^{m \times m}$ be given with $Q^3_0 \in \R^{m \times m}$ a Metzler matrix.
    Then the interconnection $F_U(G,\Phi)$ is internally stable and has $\|F_U(G,\Phi)\|_{2\to 2} < \gamma$, if there exists a positive semidefinite matrix $P$ and scalar $\gamma \ge 0$ such that $L(P,M,\gamma^2) \prec 0$ where $M$ is defined 
    as in Equations~\eqref{eq:M_ReLUMIMO}, \eqref{eq:M12MIMO} and \eqref{eq:M3MIMO}.
\end{corollary}
\vspace{0.1in}
\begin{proof}
This follows immediately from Lemmas \ref{lem:stabGeneral} and \ref{lem:dynIQCReLUMIMO}.
\end{proof}

\begin{remark}
The class of hard IQCs for repeated ReLU, obtained in Lemma~\ref{lem:dynIQCReLUMIMO}, includes the class of hard IQCs for a repeated slope-restricted nonlinearity, obtained in Lemma~\ref{lem:dynIQCSlopeMIMO}.  Let $\{Q_i\}_{i=-N}^N \subset \R^{m \times m}$ be given satisfying the conditions in Lemma~\ref{lem:dynIQCSlopeMIMO}.  We can define an IQC for ReLU as in Lemma \ref{lem:dynIQCReLUMIMO} using the following definitions:
\begin{align*}
    M_1,M_2=0, \;\; M_3=-M_0.
\end{align*}
\end{remark}

\begin{remark}
Suppose that, for some $\gamma$ and $N$, $L(P^N,M^N,\gamma^2)\prec 0$. Then the inequality $L(P^{N+1},M^{N+1},\gamma^2)\prec 0$ also holds for the same $\gamma$. Define
\begin{align*}
    P^{N+1} := \bmtx P^N & 0\\ 0 & 0\emtx, \;\; M^{N+1} := \bmtx M^N & 0 \\ 0 & 0\emtx.
\end{align*}
Therefore, any certificate feasible at horizon $N$ remains feasible at horizon $N+1$.
Increasing the filter horizon enlarges the admissible class of hard IQCs without imposing additional restrictions on the dissipation inequality.  It follows that the optimal induced-$\ell_2$ gain bound is non-increasing as a function of the filter horizon.
\end{remark}

\section{Numerical examples}

In this section, we use the condition in Corollary~\ref{cor:StabReLU} to analyze the induced-$\ell_2$ norm for $F_U(G,\Phi)$ in Figure~\ref{fig:LFTdiagram}.  The nominal part $G$ is a discrete-time, linear time-invariant (LTI) system given in \eqref{eq:LTInom} with the following transfer function:
\begin{align}
    \label{eq:GExample}
    G = \bmtx \frac{-0.13}{z-0.98} & \frac{0.21}{z-0.92} & 1 & 0 \\
    \frac{-0.3}{z-0.97} & \frac{-0.1}{z-0.91} & 0 & 1 \\
    1 & 0 & 0 & 0
    \emtx
\end{align}
The nonlinearity $F$ has $n_v=2$ inputs and $n_w=2$ outputs.  The dimensions of the disturbance and error channels are $n_d=2$ and $n_e=1$, respectively. This is a simple academic example of an RNN. 

The ``best" gain bound $\gamma$ for a given lifting horizon $N+1$ in the static QC case \cite{noori24ReLURNN}, and filter horizon $N$ in the dynamic IQC case is obtained by minimizing $\gamma^2$ subject to the LMI constraint  $L(P,M,\gamma^2)<0$, as defined in \eqref{eq:LiftedLMI}.
This minimization is a semidefinite program and was solved using CVX \cite{cvx14} as a front-end and SDPT3 \cite{toh99,tutuncu03} as the solver. 
The top row  of Table~\ref{tab:Results} shows the results as a function of the lifting horizon $N$ in filter $\Psi_N$.

\begin{table}[h!]
\centering
\scalebox{0.86}{
\begin{tabular}{ |c|c|c c c c| } 
 \hline
Nonlin.  & Analysis Type & $N=0$ & $N=1$ & $N=2$ & $N=3$\\ \hline
 \textbf{ReLU} & \textbf{Dynamic IQCs with  $\Psi_N$} & \textbf{4.017} & \textbf{1.554} & \textbf{1.300} & \textbf{1.136}\\ 
 ReLU & Static QCs with $G_{N+1}$ & 4.017 & 3.761 & 3.410 & 2.357\\  \hline \hline
 Slope Res. & Dynamic IQCs with $\Psi_N$ & 14.22 & 1.787 & 1.698 & 1.698 \\  
 Slope Res. &  Static QCs with $G_{N+1}$ & 14.22 & 8.910 & 6.473 & 3.308\\
\hline
 \end{tabular}
 }
\caption{$\ell_2$-gain upper-bounds given for ReLU and a nonlinearity that has slope restricted to $[0,1]$, using static and dynamic IQCs with different lifting and filter horizons.}
\label{tab:Results}
\end{table}

\begin{table}[h!]
\centering
\scalebox{0.86}{
\begin{tabular}{ |c|c|c c c c| } 
 \hline
Nonlin.  & Analysis Type & $N=0$ & $N=1$ & $N=2$ & $N=3$\\ \hline
\textbf{ReLU} & \textbf{Dynamic IQCs with  $\Psi_N$} & \textbf{0.4584} & \textbf{0.6230} & \textbf{0.8980} & \textbf{1.335} \\ 
ReLU & Static QCs with $G_{N+1}$ & 0.4149 & 0.3624 & 0.5707 & 0.7872 \\ \hline \hline
Slope Res. & Dynamic IQCs with $\Psi_N$ & 0.3904 & 0.4047 & 0.5668 & 0.7594\\ 
Slope Res. &  Static QCs with $G_{N+1}$ & 0.4266 & 0.3069 & 0.3616 & 0.4968\\ \hline
 \end{tabular}
 }
\caption{Comparison of average computational times (in seconds)  computed over 10 runs.}  
    \label{tab:CompTimeStabResults}
\end{table}

The case $N=0$ corresponds to the same QCs and LMI condition for the methods labelled ``dynamic IQCs" and ``static QCs".  Hence rows 2 and 3 provide the same minimal bound of 4.017.  Similarly rows 4 and 5 provide the same bound of 14.21.  It is also notable that the rows for the ReLU provide smaller bounds than the corresponding rows for slope-restricted nonlinearities. This is expected since the repeated ReLU is a special case of $[0,1]$ slope-restricted nonlinearities.  Our constraints for repeated ReLU are stronger, i.e, less conservative, than the constraints for slope restricted nonlinearities.  Finally, we note that the bounds improve (become smaller) as the lifting horizon $N$ increases.

%-----------------------------------------------------------------
\section{Conclusions}
This paper developed a new dynamic IQC for repeated ReLU nonlinearities.  We used this IQC to analyze the stability and performance of recurrent neural networks with ReLU activation functions. We showed that our new IQC provides less conservative performance bounds as compared to the class of Zames--Falb multipliers which are used for slope-restricted nonlinearities.  We also showed that the performance bounds are monotone nonincreasing with respect to the IQC filter horizon. These results were illustrated by a simple example.  Future work will consider how to incorporate more general static QCs for repeated ReLU, which are based on copositivity conditions, into this dynamic IQC framework.

%------------------------------------------------------------------
\section{Acknowledgments}

The authors acknowledge AFOSR Grant \#FA9550-23-1-0732 for funding of this work. The authors also acknowlege useful comments from Felix Biertuempfel.

\bibliographystyle{IEEEtran}
\bibliography{references} 

%------------------------------------------------------------------
\appendix 

The nominal model has the state-space representation in   \eqref{eq:LTInom}. A state space representation of the IQC filter is:
\begin{align}
\label{eq:PsiSS}
    \bmtx \psi(k+1) \\ r(k)\emtx = \bmtx A_{\psi} & B_{\psi 1} & B_{\psi 2} \\ C_{\psi} & D_{\psi 1} & D_{\psi 2} \emtx 
    \bmtx \psi(k) \\ v(k) \\ w(k)\emtx.
\end{align}
Then the augmented system
\eqref{eq:LTIAugmented} can be written with the state $\hat{x}:=\bmtx x^\top &  \psi^\top \emtx^\top$ and the following state-space matrices:
\begin{align*}
    \hat{A} &:= \bmtx A & 0 \\ B_{\psi 1}C_1 & A_{\psi}\emtx, 
    \\ 
    \hat{B}_1 & := \bmtx B_1 \\ B_{\psi 1}D_{11}+B_{\psi 2}\emtx, 
    &
    \hat{B}_2 &:= \bmtx B_2 \\ B_{\psi 1}D_{12} \emtx, 
    \\
    \hat{C}_1 & := \bmtx D_{\psi 1}C_{1} & C_{\psi}\emtx, 
    &
    \hat{C}_2 &:= \bmtx C_2 & 0\emtx,     
    \\
    \hat{D}_{11} &:= D_{\psi 1}D_{11} + D_{\psi 2},
    &
    \hat{D}_{12} & := D_{\psi 1}D_{12}, \\   
    \hat{D}_{21} & := D_{21}, 
    &
    \hat{D}_{22} & := D_{22}.
\end{align*}

\end{document}